%% file: ms.tex
\documentclass[letterpaper, 10 pt, conference]{ieeeconf}  

\IEEEoverridecommandlockouts

\usepackage{amsmath,amssymb,amsfonts}
\usepackage{bm}
\usepackage{graphicx}
\usepackage{comment}
\usepackage{algorithm}
\usepackage{algorithmic}
\usepackage{epsfig}
\usepackage{times}
\allowdisplaybreaks
\usepackage{layout}
\usepackage[dvipsnames]{xcolor}
\usepackage{siunitx}
\usepackage{fp}
\usepackage{tikz}
\usepackage{pgfplots}
\pgfplotsset{compat=newest}
\usetikzlibrary{plotmarks}
\usetikzlibrary{arrows.meta}
\usepgfplotslibrary{patchplots}
\usepackage{grffile}
\usepackage{bbm}
\usepackage{eso-pic}

\makeatletter
\let\NAT@parse\undefined
\makeatother
\usepackage{hyperref}

\pgfplotsset{compat=1.18}
\usetikzlibrary{arrows,shapes,backgrounds,patterns,fadings,decorations,
	positioning,external,arrows.meta,pgfplots.fillbetween, pgfplots.groupplots}

\newtheorem{assumption}{Assumption}
\newtheorem{theorem}{Theorem}
\newtheorem{lemma}{Lemma}
\newtheorem{remark}{Remark}
\newtheorem{definition}{Definition}

\newcommand\copyrighttext{%
	\footnotesize \textcopyright 2025 IEEE. Personal use of this material is permitted.  Permission from IEEE must be obtained for all other uses, in any current or future media, including reprinting/republishing this material for advertising or promotional purposes, creating new collective works, for resale or redistribution to servers or lists, or reuse of any copyrighted component of this work in other works.}
\newcommand\copyrightnotice{%
	\begin{tikzpicture}[remember picture,overlay]
		\node[anchor=south,yshift=5pt] at (current page.south) {\fbox{\parbox{\dimexpr\textwidth-\fboxsep-\fboxrule\relax}{\copyrighttext}}};
	\end{tikzpicture}%
}

\title{\LARGE \bf
Barrier Certificates for Unknown Systems with Latent States and Polynomial Dynamics using Bayesian Inference}

\author{Robert Lefringhausen$^{\dagger\, 1}$, Sami Leon Noel Aziz Hanna$^{\dagger\, 1}$, Elias August$^{2}$, and Sandra Hirche$^{1}$
	\thanks{*This work was supported by the European Research Council (ERC) Consolidator Grant "Safe data-driven control for human-centric systems (CO-MAN)" under grant agreement number 864686. The ChatGPT large language model \cite{openai2025} was used to improve the language and readability of the manuscript and to assist with the implementation of the simulation.}
	\thanks{$^{\dagger}$Both authors contributed equally.}
	\thanks{$^{1}$Chair of Information-oriented Control, School of Computation, Information and Technology, Technical University of Munich, Germany {\tt\small [robert.lefringhausen, sami.noel, hirche] @tum.de}.}
	\thanks{$^{2}$Department of Engineering, Reykjavík University, Iceland {\tt\small eliasaugust@ru.is}.}
}

\newcommand\DOItext{\footnotesize This is the accepted version of a paper published in the proceedings of the 2025 IEEE Conference on Decision and Control (CDC). \\ \footnotesize The final published paper can be found at \href{https://doi.org/10.1109/CDC57313.2025.11312207}{doi:10.1109/CDC57313.2025.11312207}.}

\AddToShipoutPicture{\begin{tikzpicture}[remember picture,overlay, align=center]\node[anchor=north,yshift=-15pt] at (current page.north) {\DOItext}; \end{tikzpicture}}

\begin{document}
    \maketitle
    \copyrightnotice
    \thispagestyle{empty}
    \pagestyle{empty}

    \begin{abstract}
    Certifying safety in dynamical systems is crucial, but barrier certificates\,---\,widely used to verify that system trajectories remain within a safe region\,---\,typically require explicit system models. When dynamics are unknown, data-driven methods can be used instead, yet obtaining a valid certificate requires rigorous uncertainty quantification. For this purpose, existing methods usually rely on full-state measurements, limiting their applicability. This paper proposes a novel approach for synthesizing barrier certificates for unknown systems with latent states and polynomial dynamics. A Bayesian framework is employed, where a prior in state-space representation is updated using output data via a targeted marginal Metropolis--Hastings sampler. The resulting samples are used to construct a barrier certificate through a sum-of-squares program. Probabilistic guarantees for its validity with respect to the true, unknown system are obtained by testing on an additional set of posterior samples. The approach and its probabilistic guarantees are illustrated through a numerical simulation.
    \end{abstract}

    \input{./sections/Introduction.tex}
    \input{./sections/Problem.tex}
    \input{./sections/MMH.tex}
    \input{./sections/Certificates.tex}
    \input{./sections/Simulation.tex}
    \input{./sections/Conclusion.tex}

    \newpage
    \bibliographystyle{IEEEtran}
    \bibliography{ms}

\end{document}

%% file: sections/Introduction.tex
\section{Introduction}
\label{sec:introduction}

Ensuring the safety of dynamical systems is a critical concern in applications such as human-robot interaction, autonomous driving, and medical devices, where failures can lead to severe consequences. In such scenarios, safety constraints typically mandate that the system state remains within a predefined allowable region. Barrier certificates~\cite{ames2019} provide a systematic framework for verifying safety by establishing mathematical conditions that guarantee that system trajectories remain within these regions. Traditionally, such certificates are synthesized using explicit models of the system dynamics, enabling rigorous verification through analytic methods. However, for complex systems, obtaining an accurate first-principles model can be prohibitively time-consuming or even infeasible. In cases where an explicit model is unavailable, data-driven approaches offer an alternative by leveraging observational data to infer system behavior. While promising, constructing barrier certificates from data presents significant challenges, as ensuring their validity requires rigorously accounting for uncertainties arising from limited training data (epistemic uncertainty) and noise (aleatory uncertainty).

Direct approaches construct barrier certificates from trajectory data without an intermediate model, bypassing explicit system identification \cite{salamati2024, salamati2022, nejati2023}. These methods frame safety verification as an optimization problem over collected data and provide formal probabilistic guarantees. However, a key limitation is their inability to infer system behavior outside the observed data, making them highly dependent on extensive data collection to obtain a sufficiently large number of independent samples for reliable certification. This dependency renders them impractical for systems where acquiring sufficient data is costly or infeasible.

Indirect approaches, by contrast, first infer a system model from data and then use it either to verify safety specifications or to synthesize controllers that enforce them. Bayesian learning methods, such as Gaussian process regression, provide a principled way to quantify uncertainty and propagate confidence intervals through the safety verification process \cite{dhiman2021, jackson2020, jagtap2020, zhang2024}. Beyond uncertainty quantification, Bayesian learning methods enable the integration of prior knowledge, such as structural properties or physical insights, thereby improving sample efficiency. 

However, even when a reliable model is available, finding a valid barrier certificate is inherently challenging, and systematic approaches exist only for specific system classes. For polynomial dynamics\,---\,which arise in various physical and engineering domains, including robotics and chemical reaction networks\,---\,sum-of-squares (SOS) programming \cite{papachristodoulou2005} offers a principled framework to systematically construct barrier certificates using convex optimization.

A fundamental limitation of the aforementioned direct and indirect approaches is their reliance on full-state measurements, which is often unrealistic in real-world applications where only partial or noisy measurements are available. In such cases, both the system dynamics and latent states must be inferred simultaneously, significantly complicating uncertainty quantification. While recent work has addressed uncertainty-aware optimal control for unknown systems with latent states \cite{lefringhausen2024}, to the best of our knowledge, no existing work extends such ideas to barrier certificate synthesis.

The main contribution of this paper is a novel method for identifying and verifying invariant sets for unknown systems with latent states and polynomial dynamics using only output data, addressing the fundamental limitation of existing approaches that require full-state measurements. To achieve data-efficient inference while rigorously quantifying uncertainty, we adopt a Bayesian approach and define a prior in state-space representation, which facilitates the incorporation of structural insights or domain knowledge. Since updating this prior is essential to reduce uncertainty to a level where certificate synthesis becomes feasible, but the corresponding posterior distribution is analytically intractable, we employ a targeted marginal Metropolis--Hastings sampler \cite{robert2004} to draw samples from it. These samples are then used to formulate the search for a barrier certificate and a corresponding invariant set as a sum-of-squares (SOS) optimization problem. An additional test set of posterior samples is used to assess the certificate, yielding finite-sample probabilistic guarantees that it holds for the true, unknown system. The proposed approach is illustrated in a numerical simulation.

The remainder of this paper is structured as follows. Section~\ref{sec:problem} defines the problem setting. Section~\ref{sec:MMH} introduces the marginal Metropolis--Hastings algorithm used for Bayesian inference. Section~\ref{sec:barrier_certificates} describes how the obtained samples are used to synthesize barrier certificates via SOS programming and how the validation procedure leads to probabilistic guarantees. Section~\ref{sec:simulation} illustrates the proposed method and its guarantees using a numerical simulation. Section~\ref{sec:conclusion} provides concluding remarks.

%% file: sections/Problem.tex
\section{Problem Formulation}
\label{sec:problem}
Consider a discrete-time autonomous system of the form\footnote{\textbf{Notation:} Lower/upper case bold symbols denote vectors/matrices, respectively. $\mathbb{R}$ denotes the set of real numbers and $\mathbb{N}^0$ and $\mathbb{N}$ the set of natural numbers with and without zero. $p(\cdot)$ denotes the probability density function (pdf). (Multivariate) Gaussian distributed random variables with mean $\bm{\mu}$ and variance $\bm{\Sigma}$ are denoted by $\mathcal{N}(\bm{\mu},\bm{\Sigma})$. $\bm{I}_n$ denotes the $n \times n$ identity matrix and $\bm{0}$ the zero matrix of appropriate dimension.}
\begin{subequations} \label{eq:system}
    \begin{align}
        \bm{x}_{t+1} &= \bm{f}_{\bm{\theta}}(\bm{x}_t),\\
        \bm{y}_t &= \bm{g}_{\bm{\theta}}(\bm{x}_t) + \bm{w}_t,
    \end{align}
\end{subequations}
with state $\bm{x}_t \in \mathbb{R}^{n_x}$, which is not directly observed, and output $\bm{y}_t \in \mathbb{R}^{n_y}$ at time $t\in\mathbb{N}^0$. The vector $\bm{\theta} \in \mathbb{R}^{n_\theta}$ collects all unknown parameters of the model. The system dynamics are governed by a state-transition function $\bm{f}_{\bm{\theta}}(\cdot)$ and an observation function $\bm{g}_{\bm{\theta}}(\cdot)$, which we assume to belong to a known parametric class characterized by $\bm{\theta}$. Specifically, $\bm{f}_{\bm{\theta}}(\cdot)$ is assumed to be a polynomial function with a known monomial basis, while $\bm{g}_{\bm{\theta}}(\cdot)$ is not restricted to be polynomial and may be of any form. The system is subject to independent and identically distributed (iid) measurement noise $\bm{w}_t$, drawn from a distribution $\mathcal{W}_{\bm{\theta}}$ with probability density function $p_{\mathcal{W}}(\bm{w}_t \mid \bm{\theta})$,  whose characteristics (e.g., mean, variance) are also determined by $\bm{\theta}$. The assumption of a parametric representation is motivated by the fact that, in many physical systems, the form of the dynamics can be derived from known physical laws, enabling the formulation of expressive yet interpretable models. Similarly, prior knowledge about the measurement process often guides the choice of a suitable parametric form for the observation function and noise distribution.

In order to enable theoretically justified inference and provide formal generalization guarantees, we assume the following. 
\begin{assumption}
    \label{as:prior}
    Priors $p(\bm{\theta})$ and $p(\bm{x}_{0})$ for the model parameters and the initial state of the observed trajectory are available\footnote{Note that the prior for $\bm{\theta}$ and $\bm{x}_{0}$ implies a prior for all subsequent states.}. The unknown system parameters and initial state are samples from these priors.
\end{assumption}

This assumption is justified by the fact that domain knowledge or previous experimental data can provide meaningful prior distributions in many practical settings. Similarly, the initial state is typically not arbitrary but falls within a known range based on operational conditions.

The objective of this work is to find a forward-invariant set $\mathcal{C} \subset \mathbb{R}^{n_x}$ of system~\eqref{eq:system} within a prespecified region of the state space. To approach this task in a data-driven manner, we assume that at time $t=T$, the dataset $\mathbb{D}=\{\bm{y}_{t}\}_{t=0:T}$, consisting of the most recent $T+1$ output measurements, is available. We formally define a forward-invariant set as follows.
\begin{definition}\label{def:forward_invariant}
    A set $\mathcal{C}$ is forward invariant for system \eqref{eq:system} if, whenever the state $\bm{x}_t$ enters $\mathcal{C}$, it remains in $\mathcal{C}$ for all future times, i.e.,
    \begin{equation} 
        \bm{x}_t \in \mathcal{C} \quad \Rightarrow \quad \bm{x}_\tau \in \mathcal{C} \quad \forall \tau > t. 
    \end{equation}
\end{definition}
Ensuring forward invariance is crucial in safety-critical applications, as it guarantees that the system remains within the set $\mathcal{C}$ at all times. We consider a region of the state space as safe if all trajectories starting within it remain within an allowable region\,---\,i.e., a set of states considered instantaneously admissible\,---\,for all future times. To ensure that the computed invariant set $\mathcal{C}$ is safe in this sense, we impose the constraint
\begin{equation} \label{eq:set_conditions}
    \mathcal{X}_{\text{min}} \subset \mathcal{C} \subseteq \mathcal{X}_{\text{max}}.
\end{equation}
\input{data/Safeset_Illustration}
Here, $\mathcal{X}_{\text{max}}$ is a predefined closed and compact region of allowable states. The condition $\mathcal{C} \subseteq \mathcal{X}_{\text{max}}$ restricts the search space to this region, ensuring that the computed forward-invariant set does not extend beyond $\mathcal{X}_{\text{max}}$. The nonempty, closed, and compact set $\mathcal{X}_{\text{min}}$ defines the core region of the invariant set, specifying states that must be included, such as critical initial conditions or operating points where the system is required to remain. The set hierarchy is illustrated in Figure~\ref{fig:sethierarchy}.

Due to uncertainty in the system dynamics, guaranteeing strict invariance with probability $1$ is generally infeasible. Instead, the goal is to compute a set $\mathcal{C}$ and certify that it is forward invariant with high probability. Specifically, we aim to derive guarantees of the form
\begin{equation} 
	\mathbb{P}(\bm{x}_t \in \mathcal{C} \Rightarrow \bm{x}_\tau \in \mathcal{C}\ \forall \tau > t) \geq 1-\delta,
\end{equation}
where $\delta \in (0,1)$ denotes the maximum violation probability. This probabilistic formulation enables data-driven certification despite uncertainty about the system.

To address the challenge of computing such an invariant set despite unknown dynamics and the fact that only noisy output measurements are available, we first infer the posterior distribution over the system dynamics using a marginal Metropolis--Hastings sampler. Samples from this distribution represent plausible dynamics of the unknown system and are used to synthesize a barrier certificate via a sum-of-squares (SOS) program. The resulting certificate is then assessed on an independent set of posterior samples, yielding finite-sample probabilistic guarantees that the invariance property holds for the true, unknown system.

%% file: data/Safeset_Illustration.tex
\begin{figure}[t]
	\pgfplotsset{width=6cm, compat = 1.18, 
		height = 6cm, grid= major, 
		legend cell align = left, ticklabel style = {font=\scriptsize},
		every axis label/.append style={font=\scriptsize},
		legend style = {font=\scriptsize}}		   
    	\centering
	    \begin{tikzpicture}        
		  \begin{axis}[
    			grid=none,
    			xmin=-3, xmax=3,
    			ymin=-3, ymax=3,
    			ylabel=$x_2$, xlabel=$x_1$,
    			set layers=standard,
    			reverse legend,
                    axis x line = none, 
                    axis y line = none,            
    			legend style={font=\scriptsize, at={(1,1)},anchor=north east, row sep=2pt}, ylabel shift = -6 pt]                     
                \node[circle, draw, text=black, thick, gray, color=black, fill = lightgray!70, minimum size = 3.75cm] (c) at (0,0){};           
                \def\file{data/C_illustration.txt}	
                \addplot[dashed,very thick, OrangeRed] table[x=x_1,y=x_2]{\file};
                \node[circle,draw,minimum size = 1.5cm, thick, text=black,color=OliveGreen, fill = OliveGreen!20] (c) at (0,0){$\mathcal{X}_\text{min}$};           
                \node[text width = 3cm, text = OrangeRed] (label) at (2.4cm,2.8cm) {$\mathcal{C}$};
                \node[text width = 3cm, text = black] (label) at (3.5cm,3.5cm) {$\mathcal{X}_\text{max}$};    	
		  \end{axis}          
	   \end{tikzpicture}
    \vspace{-0.1cm}
    \caption{Illustration of the constraint $\mathcal{X}_{\text{min}} \subset \mathcal{C} \subseteq \mathcal{X}_{\text{max}}$. The outer set $\mathcal{X}_{\text{max}}$ defines the allowable region for the forward-invariant set $\mathcal{C}$, while the inner set $\mathcal{X}_{\text{min}}$ specifies states that must be included in $\mathcal{C}$.}
    \label{fig:sethierarchy}
    \vspace*{-0.5cm}
\end{figure}

%% file: sections/MMH.tex
\section{Marginal Metropolis--Hastings Sampling}
\label{sec:MMH}
In order to derive barrier certificates, it is essential to reduce uncertainty about the system\,---\,particularly the state transition function $\bm{f}_{\bm{\theta}}(\cdot)$\,---\,by incorporating information from the observed data $\mathbb{D}$. If the prior distribution exhibits high variance, the resulting uncertainty in the model parameters would otherwise prevent the derivation of barrier certificates or lead to overly conservative ones. In the Bayesian framework, uncertainty is reduced by inferring the posterior distribution $p(\bm{\theta} \mid \mathbb{D})$, which combines prior knowledge with evidence from the observed data. However, in the setting considered here, this posterior is intractable due to its dependence on the latent, i.e., not directly observable, state trajectory $\bm{x}_{0:T}$. This introduces a high-dimensional integral over possible state trajectories that cannot be solved in closed form \cite{andrieu2010}. To address this challenge, Markov chain Monte Carlo methods \cite{robert2004} can be used to draw samples from the joint posterior distribution of the parameters and latent states. This section introduces the sampling-based inference procedure used to obtain these posterior samples.

\begin{algorithm}[t]
	\caption{Marginal Metropolis Hastings sampler}\label{alg:MMH_sampler}
	\begin{algorithmic}[1]
		\renewcommand{\algorithmicrequire}{\textbf{Input:}}
		\renewcommand{\algorithmicensure}{\textbf{Output:}}
		\REQUIRE Observations $\mathbb{D}$, model $\bm{f}_{\bm{\theta}}(\cdot)$, $\bm{g}_{\bm{\theta}}(\cdot)$, $p_{\mathcal{W}}(\cdot \mid \bm{\theta})$, priors $p(\bm{\theta})$, $p(\bm{x}_{0})$, proposal distribution $q(\bm{\theta}', \bm{x}_0' \mid \bm{\theta}^{[k]}, \bm{x}_0^{[k]})$, initial values $\bm{\theta}^{[1]}$,  $\bm{x}_{0}^{[1]}$, number of samples $K$, burn-in period $K_b$, thinning parameter $k_d$
		\ENSURE $K$ samples from $p(\bm{\theta}, \bm{x}_{0:T} \mid \mathbb{D})$
		\STATE Initialize: $k \gets 1$
		\STATE Get $\bm{x}_{0:T}^{[1]}$ and compute likelihood using \eqref{eq:likelihood}
		\WHILE {$k \leq K_b + 1 + (K-1) (k_d+1)$}
		\STATE Propose $\bm{\theta}', \bm{x}_0' \sim q(\bm{\theta}', \bm{x}_0' \mid \bm{\theta}^{[k]}, \bm{x}_0^{[k]})$
		\STATE Perform forward simulation to obtain $\bm{x}_{0:T}'$
		\STATE Compute likelihood $p(\mathbb{D} \mid \bm{\theta}', \bm{x}_0')$ using \eqref{eq:likelihood}
		\STATE Compute acceptance ratio $\alpha$ using \eqref{eq:acceptance_probability}
		\IF{$u \sim \text{Uniform}(0, 1) < \alpha$}
		\STATE Accept proposal: $k \gets k + 1$, $\bm{\theta}^{[k]} \gets \bm{\theta}', \bm{x}_0^{[k]} \gets \bm{x}_0'$
		\ENDIF
		\ENDWHILE
		\STATE Discard burn-in and perform thinning
	\end{algorithmic}
\end{algorithm}

We propose using a marginal Metropolis--Hastings (MMH) sampler, as outlined in Algorithm~\ref{alg:MMH_sampler}, to draw samples from the joint posterior distribution $p(\bm{\theta}, \bm{x}_{0:T} \mid \mathbb{D})$. The general idea of MMH sampling is to propose new values for the uncertain quantities and either accept or reject them based on their likelihood, which reflects how well they explain the observed data. Since the latent state trajectory $\bm{x}_{0:T}$ is uniquely determined by the initial state $\bm{x}_0$ and the model parameters $\bm{\theta}$, it is sufficient to propose new values only for $\bm{\theta}$ and $\bm{x}_0$. These proposals are generated based on the last accepted sample $(\bm{\theta}^{[k]}, \bm{x}_0^{[k]})$ using a proposal distribution with pdf $q(\bm{\theta}', \bm{x}_0' \mid \bm{\theta}^{[k]}, \bm{x}_0^{[k]})$, which must cover the support of the prior distributions $p(\bm{\theta})$ and $p(\bm{x}_0)$ (line~4). To evaluate the likelihood of a proposed pair $(\bm{\theta}', \bm{x}_0')$, the corresponding state trajectory $\bm{x}_{0:T}'$ is determined by performing a forward simulation with the deterministic state transition function $\bm{f}_{\bm{\theta}'}(\cdot)$ starting from the initial state $\bm{x}_0'$ (line~5). The resulting state trajectory $\bm{x}_{0:T}'$ is then used to compute the likelihood of the observations (line~6) as
\begin{equation}\label{eq:likelihood}
	p\left( \mathbb{D}=\bm{y}_{0:T} \mid \bm{\theta}', \bm{x}_0' \right) = \prod_{t=0}^{T} p_{\mathcal{W}}\left(\bm{y}_t - \bm{g}_{\bm{\theta}'}(\bm{x}_t') \mid \bm{\theta}'\right).
\end{equation}
The acceptance probability $\alpha$ for a proposed pair $(\bm{\theta}', \bm{x}_0')$ (line~7) is given by \cite{robert2004}
\begin{equation}\label{eq:acceptance_probability}
    \min\left(1, \frac{p(\mathbb{D} \mid \bm{\theta}', \bm{x}_0') p(\bm{\theta}', \bm{x}_0') q(\bm{\theta}^{[k]}, \bm{x}_0^{[k]} \mid \bm{\theta}', \bm{x}_0')}{p(\mathbb{D} \mid \bm{\theta}^{[k]}, \bm{x}_0^{[k]}) p(\bm{\theta}^{[k]},\bm{x}_0^{[k]}) q(\bm{\theta}', \bm{x}_0' \mid \bm{\theta}^{[k]}, \bm{x}_0^{[k]})}\right).
\end{equation}
This expression compares the likelihoods of the proposed and current values, weighted by their prior probabilities and the proposal distributions. This acceptance ratio ensures that proposals that better explain the observed data are more likely to be accepted. If a proposed pair $(\bm{\theta}', \bm{x}_0')$ is accepted, it replaces the current state in the Markov chain, and the new parameters are used as the basis for subsequent proposals (lines~8-10). Otherwise, the chain remains at the current state.

It can be shown under mild assumptions that the invariant distribution of the MMH sampler\,---\,i.e., the distribution to which the Markov chain converges asymptotically\,---\,is $p(\bm{\theta}, \bm{x}_{0:T} \mid \mathbb{D})$ \cite{tierney1994}. This ensures that the samples generated by the MMH sampler represent the true posterior distribution, provided that a sufficient number of samples are drawn. Since the first samples depend heavily on the initialization and may not accurately represent the target distribution, the first $K_b$ samples are discarded (burn-in period). Moreover, consecutive samples tend to be correlated, which is undesirable for the intended application. To mitigate this correlation and obtain approximately independent samples, $k_d$ samples are discarded between every pair of accepted samples. This approach, known as thinning, is a straightforward and effective way to reduce autocorrelation. The parameter $k_d$ should be chosen to balance efficiency and independence\,---\,small enough to avoid excessive discarding of samples but large enough to ensure low autocorrelation among the retained ones. 

\begin{remark}\label{rem:staged_sampler}
Directly sampling from the full posterior distribution can be challenging, particularly when the distribution is sharply concentrated or exhibits strong correlations between parameters. This issue frequently arises in state-space models, where parameters and latent states jointly influence the observations and are thus often strongly correlated. In such cases, random walk proposals tend to be inefficient, frequently resulting in low acceptance rates and high sample autocorrelation. An effective strategy to mitigate this problem is staged inference, in which the number of data points used in the likelihood estimation~\eqref{eq:likelihood} is gradually increased. At each stage, the proposal distribution~$q(\cdot)$ can be adapted based on the empirical covariance of the current samples, allowing the sampler to progressively track the evolving posterior geometry. This approach improves convergence and reduces sensitivity to the initial values~$\bm{\theta}^{[1]}$ and~$\bm{x}_0^{[1]}$. When such staged adaptation is restricted to the burn-in phase and the proposal distribution remains fixed during the final sampling, the theoretical convergence guarantees of the sampler are preserved.
\end{remark}

With an appropriately long burn-in period and effective thinning, the proposed MMH sampler produces nearly uncorrelated samples from the target distribution, which justifies the following assumption.

\begin{assumption}
\label{as:iid_posterior}
    The MMH sampler provides $K$ independent samples $\{\bm{\theta},\bm{x}_{0:T}\} ^{[1:K]}$ from the distribution $p(\bm{\theta}, \bm{x}_{0:T} \mid \mathbb{D})$.
\end{assumption}

Although the MMH sampler provides samples $\{ \bm{\theta}, \bm{x}_{0:T} \}^{[1:K]}$ from the joint posterior distribution over parameters and latent state trajectories, only the parameter samples are relevant for determining forward-invariant sets. In the remainder of this paper, we therefore consider only the set of parameter samples $\mathbb{S} = \{\bm{\theta}^{[1]},\dots,\bm{\theta}^{[K]}\}$, effectively marginalizing out the latent state trajectories.

%% file: sections/Certificates.tex
\section{Synthesis and Validation of Barrier Certificates}
\label{sec:barrier_certificates}
In this section, we show how barrier certificates with probabilistic performance guarantees for system~\eqref{eq:system} can be derived from the iid samples of the posterior distribution of model parameters generated by the MMH sampler.  
These samples represent possible realizations of the system dynamics and collectively capture the remaining uncertainty after updating the prior with observations $\mathbb{D}$\,---\,if the samples are highly diverse, the uncertainty is large; if they are similar, the uncertainty is small. For the purpose of barrier certificate synthesis and validation, the set of samples $\mathbb{S} = \{\bm{\theta}^{[1]}, \dots, \bm{\theta}^{[K]}\}$ is divided into two disjoint subsets: a training set $\mathbb{S}_{\text{train}}$ of size $K_{\text{train}}$, used to synthesize a barrier certificate, and a test set $\mathbb{S}_{\text{test}}$ of size $K_{\text{test}}=K-K_{\text{train}}$, used to validate the candidate and derive probabilistic guarantees. Section~\ref{subsec:general_idea} introduces the general idea of finding an invariant set using barrier function conditions. Section~\ref{subsec:SOS_implementation} explains how these conditions are implemented using sum-of-squares (SOS) programming to obtain a prospective barrier certificate based on the training set. Finally, Section~\ref{subsec:guarantees} derives finite-sample probabilistic guarantees for invariance based on the test set.

\subsection{Barrier Conditions for Invariant Sets}\label{subsec:general_idea}
The general idea of our approach is to find a continuous function $h:\mathbb{R}^{n_x} \rightarrow \mathbb{R}$ (the barrier function) whose zero superlevel set defines a forward-invariant region of the state space. Formally, consider any compact set $\mathcal{C}$ satisfying
\begin{subequations} \label{eq:invariant_set_condition}
    \begin{align}
        h(\bm{x}) &= 0 \quad \forall \bm{x} \in \partial \mathcal{C}, \\
        h(\bm{x}) &> 0 \quad \forall \bm{x} \in \operatorname{Int}(\mathcal{C}),
    \end{align}
\end{subequations} 
where $\partial \mathcal{C}$ denotes the boundary and $\operatorname{Int}(\mathcal{C})$ the interior of the set. If $h(\cdot)$ further satisfies the inequality
\begin{equation} \label{eq:barrier_condition}
    h(\bm{f}_{\bm{\theta}}(\bm{x})) - h(\bm{x}) + \alpha(h(\bm{x})) \geq 0
\end{equation}
for all $\bm{x} \in \mathcal{C}$, where $\alpha: [0,\infty) \rightarrow \mathbb{R}_{\geq 0}$ is a class $\mathcal{K}$ function\,---\,i.e., strictly increasing with $\alpha(0) = 0$\,---\,that additionally satisfies $\alpha(r) < r$ for all $r>0$, then $\mathcal{C}$ is forward invariant for the system~\eqref{eq:system} \cite[Theorem 1]{ahmadi2019}. This means that once a trajectory enters $\mathcal{C}$, it remains confined. The function $h(\cdot)$ serves as a barrier function, as it makes sure that trajectories do not cross the boundary of $\mathcal{C}$ from within\,---\,effectively repelling them from the border as dictated by the barrier function condition. Note that \eqref{eq:barrier_condition} alone does not guarantee the existence of a forward-invariant set $\mathcal{C}$ satisfying \eqref{eq:invariant_set_condition}. In principle, $h(\cdot)$ could be negative everywhere with $\mathcal{C}=\emptyset$. Therefore, we must explicitly enforce the nonemptiness of~$\mathcal{C}$.

While \eqref{eq:invariant_set_condition} and \eqref{eq:barrier_condition} provide sufficient conditions for a forward-invariant set, they cannot be utilized directly in our setting since the true state transition function $\bm{f}_{\bm{\theta}}(\cdot)$ is unknown. Instead, we aim to find a set $\mathcal{C}$ that is forward invariant for all samples in the training set. Specifically, we aim to find a continuous function $h(\cdot)$ such that
\begin{equation} \label{eq:sample_barrier_condition}
    h(\bm{f}_{\tilde{\bm{\theta}}}(\bm{x})) - h(\bm{x}) + \alpha(h(\bm{x})) \geq 0 \ \forall \bm{x} \in \mathcal{X}_{\text{max}},\ \forall \tilde{\bm{\theta}} \in \mathbb{S}_{\text{train}}.
\end{equation}
The general idea is that with a sufficient number of samples, a robust certificate can be expected that is likely to hold for the unknown dynamics \eqref{eq:system}.

\begin{remark}
While, in theory, the inequality \eqref{eq:barrier_condition} is required to hold only on the (unknown) set $\mathcal{C}$, enforcing it solely on this set would require explicit knowledge of $\mathcal{C}$, which leads to an intractable problem. Approaches such as the one presented in \cite{schneeberger2023} address this challenge using alternating optimization schemes, which require a good initial guess for the unknown quantities\,---\,information that is unavailable in our setting due to uncertainty in the dynamics. Therefore, in \eqref{eq:sample_barrier_condition}, we require the inequality to hold on the known superset $\mathcal{X}_{\text{max}}$, introducing some conservatism but retaining tractability.
\end{remark}

Next, we describe how these conditions are translated into a computational framework based on sum-of-squares (SOS) programming.

\subsection{Sum-of-Squares Implementation}\label{subsec:SOS_implementation}
In order to formalize the search for a valid barrier function $h(\bm{x})$ satisfying \eqref{eq:sample_barrier_condition}, we employ sum-of-squares (SOS) programming \cite{papachristodoulou2005}, a convex optimization framework that allows for the efficient verification of polynomial inequalities using semidefinite programming. A polynomial function $p(\bm{x})$ is a sum-of-squares (``is SOS'') if there exist polynomials $q_i(\bm{x})$ such that $p(\bm{x}) = \sum_i q_i^2(\bm{x})$ which guarantees non-negativity of $p(\bm{x})$, i.e., $p(\bm{x})\geq 0$ for all $ \bm{x} \in \mathbb{R}^{n_x}$. The key advantage of SOS programming is that it provides a tractable method to certify polynomial inequalities\,---\,such as the barrier function condition \eqref{eq:barrier_condition}\,---\,by formulating them as SOS constraints. Specifically, we assume $h(\cdot)$ is a polynomial function of a given degree. For $\alpha(\cdot)$, we choose the polynomial function $\alpha(\bm{x}) = \gamma \bm{x}$ with $\gamma \in (0,1)$. This allows us to express \eqref{eq:sample_barrier_condition} as the following SOS constraint
\begin{equation} \label{eq:SOS_barrier_condition}
h(\bm{f}_{\tilde{\bm{\theta}}}(\bm{x}))-h(\bm{x})+\gamma h(\bm{x})\text{ is SOS } \forall \tilde{\bm{\theta}} \in \mathbb{S}_{\text{train}}.
\end{equation}
However, this formulation imposes the inequality in \eqref{eq:sample_barrier_condition} over the entire state space, i.e., for all $\bm{x} \in \mathbb{R}^{n_x}$. When the barrier function $h(\cdot)$ has a nonempty and compact zero superlevel set $\mathcal{C}$, verifying this inequality globally is significantly more challenging than doing so locally. Thus, directly enforcing \eqref{eq:SOS_barrier_condition} is overly restrictive and may lead to infeasibility.

To relax the SOS constraint \eqref{eq:SOS_barrier_condition} to only hold within the admissible region $\mathcal{X}_{\text{max}}$, we assume knowledge of polynomial functions $r_\text{min},r_\text{max}: \mathbb{R}^{n_x} \rightarrow \mathbb{R}$, which are nonnegative on $\mathcal{X}_\text{min}$ and $\mathcal{X}_\text{max}$, respectively. Ideally, $r_\text{min}(\cdot)$ and $r_\text{max}(\cdot)$ are designed to become negative outside the respective regions. A typical choice for a circular set with radius $\rho$ is $r(\bm x) := \rho^2-\bm x^{\mathsf T}\bm x$. Assuming knowledge of these functions is not restrictive in practice, as polynomial functions are highly expressive. To enforce the condition \eqref{eq:sample_barrier_condition} only within the admissible region $\mathcal{X}_{\text{max}}$, we adopt a standard relaxation and modify the SOS constraint \eqref{eq:SOS_barrier_condition} to
\begin{equation} \label{eq:SOS_barrier_condition_relaxed}
h(\bm{f}_{\tilde{\bm{\theta}}}(\bm{x}))-h(\bm{x})+\gamma h(\bm{x}) - q_1(\bm{x})r_\text{max}(\bm{x})\text{ is SOS } \forall \tilde{\bm{\theta}} \in \mathbb{S}_{\text{train}},
\end{equation}
where $q_1:\mathbb{R}^{n_x} \rightarrow \mathbb{R}$ is an SOS polynomial of a given degree, ensuring $q_1(\bm{x})\geq 0\ \forall\bm{x} \in \mathbb{R}^{n_x}$, and $r_\text{max}(\bm x)\geq 0\ \forall\bm x\in \mathcal{X}_\text{max}$. Compared to the original constraint \eqref{eq:SOS_barrier_condition}, this relaxed formulation enforces the inequality \eqref{eq:sample_barrier_condition} only within $\mathcal{X}_{\text{max}}$ while allowing flexibility outside this region. The additional SOS polynomial $q_1(\cdot)$ acts as a non-negative multiplier, increasing the set of feasible solutions and improving the solvability of the overall SOS program.

To enforce the existence of a forward-invariant set $\mathcal{C}$ that satisfies \eqref{eq:set_conditions}, i.e., $\mathcal{X}_{\text{min}} \subset \mathcal{C} \subseteq \mathcal{X}_{\text{max}}$, we explicitly incorporate this constraint into the SOS formulation. For this, we assume that the sets $\mathcal{X}_{\text{min}}$ and $\mathcal{X}_{\text{max}}$ can be represented using polynomial functions. Specifically, we assume knowledge of polynomial functions $h_{\text{min}}: \mathbb{R}^{n_x} \rightarrow \mathbb{R}$ and $h_{\text{max}}: \mathbb{R}^{n_x} \rightarrow \mathbb{R}$ satisfying
\begin{subequations} \label{eq:set_polynomials}
    \begin{align}
        h_{\text{min}}(\bm x) &= 0  \quad \forall \bm{x} \in \partial \mathcal{X}_{\text{min}}, \label{eq:set_polynomials_2}\\
        h_{\text{max}}(\bm x) &= 0  \quad \forall \bm{x} \in \partial \mathcal{X}_{\text{max}}.\label{eq:set_polynomials_3}
    \end{align}
\end{subequations}
Again, this assumption is not restrictive as polynomial functions are highly expressive, and there is typically some flexibility in choosing $\mathcal{X}_{\text{min}}$ and $\mathcal{X}_{\text{max}}$ as they merely define the search space. The following lemma provides a constructive way to enforce that the barrier function has a nonempty zero superlevel set $\mathcal{C}$ satisfying~\eqref{eq:set_conditions}, using polynomial inequalities involving $h_{\text{min}}(\cdot)$ and $h_{\text{max}}(\cdot)$.
\begin{lemma}\label{lem:existence_conditions}
    If there exists a polynomial function $h(\cdot)$ satisfying the inequality \eqref{eq:barrier_condition} for all $\bm{x} \in \mathcal{X}_{\text{max}}$, and polynomial functions $h_{\text{min}}(\cdot)$ and $h_{\text{max}}(\cdot)$ satisfying \eqref{eq:set_polynomials}, such that
    \begin{subequations}\label{eq:existence_condition}
        \begin{align}
             &h(\bm{x})>h_{\text{min}}(\bm{x}) \quad \forall \bm{x} \in \mathcal{X}_{\text{min}},\label{eq:emptyset_condition}\\
            &h(\bm x)\leq h_{\text{max}}(\bm{x}) \quad \forall \bm{x} \in \mathcal{X}_{\text{max}},\label{eq:boundedness_condition}
        \end{align}
    \end{subequations}
    then, there exists the set 
    \begin{equation}
    \mathcal{C}:=\{x \in \mathcal{X}_{\text{max}} \mid h(\bm{x})\geq 0\} \cup \mathcal{X}_{\text{min}}
    \end{equation}
    that is forward invariant and satisfies $\mathcal{X}_\text{min}\subset \mathcal{C} \subseteq \mathcal{X}_{\text{max}}$. 
\end{lemma}
\begin{proof}
If the inequality \eqref{eq:barrier_condition} is satisfied for all $\bm{x} \in \mathcal{X}_{\text{max}}$, then any compact subset of $\mathcal{X}_{\text{max}}$ fulfilling \eqref{eq:invariant_set_condition} is forward invariant \cite[Theorem 1]{ahmadi2019}. Condition \eqref{eq:boundedness_condition} combined with \eqref{eq:set_polynomials_3} prevents any set $\mathcal{C}$ intersecting $\partial\mathcal{X}_\text{max}$ from extending beyond $\mathcal{X}_\text{max}$, as this would contradict \eqref{eq:invariant_set_condition}. Similarly, condition \eqref{eq:emptyset_condition} along with \eqref{eq:set_polynomials_2} guarantees that $h(\bm{x})>0\ \forall\bm x \in \partial \mathcal{X}_{\text{min}}$. By continuity of the polynomial function $h(\cdot)$, this implies that there exists a set $\tilde{\mathcal{C}} \subseteq \mathcal{X}_{\text{max}}$, containing  $\partial \mathcal{X}_{\text{min}}$ and satisfying \eqref{eq:invariant_set_condition}.
Since $\tilde{\mathcal{C}}$ is forward invariant and contains $\partial\mathcal{X}_{\text{min}}$, the union $\mathcal{C}:= \mathcal{X}_{\text{min}} \cup \tilde{\mathcal{C}}$ is also forward invariant and satisfies $\mathcal{X}_\text{min}\subset \mathcal{C} \subseteq \mathcal{X}_{\text{max}}$.
\end{proof}

Since the inequalities \eqref{eq:existence_condition} can be naturally expressed as SOS constraints, Lemma~\ref{lem:existence_conditions} enables us to enforce the existence of a forward-invariant set satisfying \eqref{eq:set_conditions} within the SOS framework. Similarly to before, we employ the relaxation functions $r_\text{min}(\cdot)$ and $r_\text{max}(\cdot)$ to ensure that the inequalities are only enforced locally. We thus obtain the following SOS program: 
\begin{subequations} \label{eq:SOS_program}
    \begin{align}
        \text{find}\ h(\cdot),\ q_1(\cdot),\ q_2(\cdot),\ q_3(\cdot), \ \text{s.t.} \ \forall \tilde{\bm{\theta}}& \in \mathbb{S}_{\text{train}},\nonumber\\
        h(\bm{f}_{\tilde{\bm{\theta}}}(\bm{x}))-h(\bm{x})+\gamma h(\bm{x}) - q_1(\bm{x})r_\text{max}(\bm{x})&\text{ is SOS}, \label{eq:SOS_program_1}\\
        h(\bm{x}) - h_{\text{min}}(\bm{x}) - \epsilon - q_2(\bm{x})r_\text{min}(\bm{x})&\text{ is SOS}, \label{eq:SOS_program_2} \\
        h_{\text{max}}(\bm{x}) - h(\bm{x}) - q_3(\bm{x})r_\text{max}(\bm{x})&\text{ is SOS}, \label{eq:SOS_program_3}
    \end{align}
\end{subequations}
where $\epsilon>0$ is a small constant. This constant is necessary because SOS constraints certify non-negativity but not strict positivity. The SOS feasibility problem \eqref{eq:SOS_program} can then be solved using semidefinite programming, yielding the coefficients of the polynomial functions $h(\cdot)$, $q_1(\cdot)$, $q_2(\cdot)$, and $q_3(\cdot)$. If a feasible solution is found, then the nonempty set
\begin{equation}\label{eq:set_C}
\mathcal{C}:=\{ \bm{x} \in \mathcal{X}_{\text{max}}\mid h(\bm{x}) \geq 0\} \cup \mathcal{X}_{\text{min}}
\end{equation} 
that is forward invariant for all samples in the training set is guaranteed to exist.
\begin{remark}
    Infeasibility of the SOS program~\eqref{eq:SOS_program} does not necessarily imply that the true system lacks a forward-invariant set within the admissible region. A fundamental limitation of the proposed approach is that the SOS framework can only capture invariant sets that admit a barrier function $h(\cdot)$ representable by a finite-degree polynomial, and some sets may lie outside this representation class. Even if such a polynomial barrier function exists in principle, the chosen degree may be insufficient to capture the invariant set. In practice, the degree should therefore be selected as high as computational resources allow, since higher degrees increase expressiveness without introducing a risk of overfitting, due to the large number of sample-based constraints that must be satisfied. Moreover, unsuitable structural choices for $h_{\text{min}}(\cdot)$ or $h_{\text{max}}(\cdot)$ may overly restrict the feasible search space and thereby lead to infeasibility. Finally, infeasibility may also result from high uncertainty about the unknown system: if the posterior samples vary significantly, it may be impossible to find a single $h(\cdot)$ satisfying the SOS condition~\eqref{eq:SOS_program_1} for all of them. In this case, infeasibility reflects the conservatism of enforcing invariance across all training samples. This conservatism can be alleviated by reducing the training set size, while the underlying uncertainty can be mitigated by collecting more data~$\mathbb{D}$ or adopting a more informative prior distribution.
\end{remark}
\begin{remark}\label{rem:circular_searchspace}
    In certain cases, it is possible to include the search for suitable functions $h_{\text{min}}(\cdot)$ and $h_{\text{max}}(\cdot)$ directly within the SOS program \eqref{eq:SOS_program}. Since these functions define the admissible region for $h(\cdot)$ via constraints \eqref{eq:SOS_program_2} and \eqref{eq:SOS_program_3}, allowing them to be partially optimized increases the space of feasible solutions and can significantly improve the solvability of the SOS program. To this end, one must define a structural form for $h_{\text{min}}(\cdot)$ and $h_{\text{max}}(\cdot)$ that guarantees satisfaction of the conditions in \eqref{eq:set_polynomials}, without fixing the functions entirely. This is possible in certain structured cases. For example, if the sets $\mathcal{X}_{\text{min}}$ and $\mathcal{X}_{\text{max}}$ are Euclidean balls centered at the origin with radii $\rho_{\text{min}}$ and $\rho_{\text{max}}$, respectively, the functions can be defined as 
    \begin{subequations} \label{eq:circular_searchspace}
    \begin{align}
        h_{\text{min}}(\bm{x}) &:= q_4(\bm{x}^{\mathsf T}\bm{x}) - q_4(\rho_{\text{min}}^2) \\
        h_{\text{max}}(\bm{x}) &:= q_5(\bm{x}^{\mathsf T}\bm{x}) - q_5(\rho_{\text{max}}^2)
    \end{align}
    \end{subequations}
    where $q_4(\cdot)$ and $q_5(\cdot)$ are univariate polynomials of a chosen degree. The structure in \eqref{eq:circular_searchspace} guarantees that both functions are zero on the boundary of $\mathcal{X}_{\text{min}}$ and $\mathcal{X}_{\text{max}}$, respectively, while providing the solver with the flexibility to shape the functions through the coefficients of $q_4(\cdot)$ and $q_5(\cdot)$, thereby improving the feasibility of the overall program. Analogously, if $\mathcal{X}_{\text{min}}$ or $\mathcal{X}_{\text{max}}$ are ellipsoids, the search for $h_{\text{min}}(\cdot)$ or $h_{\text{max}}(\cdot)$ can be incorporated into the SOS program by setting $\rho=1$ and replacing $\bm{x}^{\mathsf T}\bm{x}$ with $\bm{x}^{\mathsf T}\bm{Q}\bm{x}$, where $\bm{Q}$ is a positive definite matrix that captures the size and orientation of the respective set.
\end{remark}

\subsection{Validation and Guarantees}\label{subsec:guarantees}
In the following, we derive probabilistic guarantees that the barrier function candidate obtained by solving the SOS program~\eqref{eq:SOS_program} also holds for the true, unknown system. To assess generalization beyond the samples in the training set, we evaluate the candidate on a separate test set $\mathbb{S}_{\text{test}}$ of size $K_{\text{test}}$. Since only constraint~\eqref{eq:SOS_program_1} depends on the system parameters, it is re-evaluated for each sample in the test set. The barrier function candidate is valid for a sample $\tilde{\bm{\theta}}$ if
\begin{equation} \label{eq:validation}
    h\left(\bm{f}_{\tilde{\bm{\theta}}}(\bm{x})\right) - h(\bm{x}) + \gamma h(\bm{x}) - q_1(\bm{x}) r(\bm{x}) \text{ is SOS}.
\end{equation}
The following theorem establishes probabilistic guarantees on forward invariance of the set~\eqref{eq:set_C}, induced by the barrier function candidate, for the true, unknown system~\eqref{eq:system}.

\begin{theorem}\label{theo:robustness_guarantees}
    Suppose that the SOS program \eqref{eq:SOS_program} is feasible, yielding a barrier function candidate $h(\cdot)$, and that Assumptions~\ref{as:prior} and \ref{as:iid_posterior} hold. Let $V$ be the number of violations of condition~\eqref{eq:validation} observed among the $K_{\text{test}}$ samples in $\mathbb{S}_{\text{test}}$, and let $\beta \in (0,1)$ be a chosen confidence level. Then, with confidence at least $1-\beta$, the associated set $\mathcal{C}$, defined by~\eqref{eq:set_C}, is forward invariant for the true, unknown system~\eqref{eq:system} with probability at least $1 - \delta$, where $\delta \in (0,1)$ is the unique solution to
    \begin{equation}\label{eq:clopper-pearson}
        \sum_{i=0}^V \binom{K_{\text{test}}}{i}\delta^i(1-\delta)^{K_{\text{test}}-i}=\beta.
    \end{equation}
\end{theorem}
\begin{proof}
    Assumption~\ref{as:iid_posterior} ensures that the samples in $\mathbb{S}_{\text{test}}$ are drawn independently from the posterior distribution $p(\bm{\theta}\mid\mathbb{D})$. By Assumption~\ref{as:prior}, the true, unknown system dynamics are also distributed according to this posterior and thus represent a potential (unseen) sample. Each test can therefore be interpreted as an independent Bernoulli trial with some unknown violation probability $\nu$, i.e., the probability that $h(\cdot)$ fails to satisfy condition~\eqref{eq:validation} under a posterior draw. Consequently, the number of observed violations $V$ among $K_{\text{test}}$ trials follows a binomial distribution with parameter $\nu$, i.e., $V \sim \mathrm{Binomial}(K_{\text{test}}, \nu)$.

    Based on this binomial model, the Clopper--Pearson upper $(1-\beta)$ confidence bound on $\nu$ \cite{clopper1934}, given by the unique solution $\delta$ of~\eqref{eq:clopper-pearson}, guarantees that with confidence at least $1-\beta$ the true violation probability satisfies $\nu \le \delta$. Accordingly, with the same confidence, the set $\mathcal{C}$ is forward invariant for the true system with probability at least $1-\delta$, which completes the proof.
\end{proof}

The lower bound on the probability that the set $\mathcal{C}$, induced by the barrier function candidate, is forward invariant for the true, unknown system, as established by Theorem~\ref{theo:robustness_guarantees}, depends on three factors: the number of test samples $K_{\text{test}}$, the number of observed violations $V$, and the confidence parameter $\beta$. For a fixed violation rate $V/K_{\text{test}}$, the lower bound $1-\delta$ increases as the number of test samples grows. Conversely, for a fixed test set size, observing fewer violations $V$ naturally yields a stronger guarantee \cite{clopper1934}.

Intuitively, the confidence parameter $\beta$ quantifies the probability of drawing an unrepresentative test set from the posterior distribution. Interpreted in the frequentist sense, this means that if the entire procedure (data collection, synthesis of $h(\cdot)$, and testing on $\mathbb{S}_{\text{test}}$) were to be repeated many times, the bound $1-\delta$ would hold in at least a $(1-\beta)$ fraction of repetitions. While essential for theoretical validity, the influence of $\beta$ is mild, since changing $\beta$ only slightly shifts the relevant Beta-distribution quantiles \cite{thulin2014}. Thus, in practice $\beta$ can be chosen very small (e.g., $\beta = 10^{-5}$), as the resulting change in $1-\delta$ is small compared to the effects of $K_{\text{test}}$ and $V$.

If the lower bound $1-\delta$ from Theorem~\ref{theo:robustness_guarantees} exceeds a desired threshold, the candidate function can be accepted. Otherwise, the entire procedure may be repeated with an enlarged training set. Since the barrier certificate must be valid for all samples in the training set, increasing the training set enforces validity over a larger portion of the posterior support and thereby enhances robustness. When repeating the procedure, however, the allocation of the confidence parameter $\beta$ must be handled carefully, as each repetition consumes confidence budget and requires an appropriate $\beta$-spending strategy \cite{hochberg1987}.

%% file: sections/Simulation.tex
\section{Simulation}
\label{sec:simulation}
In this section, we illustrate the proposed approach and its probabilistic guarantees using a numerical simulation\footnote{The code is available at \href{https://github.com/TUM-ITR/BABI}{\texttt{github.com/TUM-ITR/BABI}}.}. We consider a discretized FitzHugh–Nagumo system, a well-known two-dimensional nonlinear model originally developed to describe neuronal dynamics \cite{cebrian2024}. The system dynamics are given by
\begin{equation} \label{eq:fitzhugh-nagumo}
    f(\bm{x}_{t}, \bm{\theta}) = \bm{x}_t + \Delta t \cdot \begin{bmatrix}  
        x_{1,t} - x_{2,t} - x_{1,t}^3 \\  
        \theta_1 x_{1,t} + \theta_2 x_{2,t} + \theta_3
    \end{bmatrix},
\end{equation} 
where  $\Delta t = 0.1$ is the discretization time step and $\bm{\theta} = [0.01, -0.005, 0]^{\mathsf T}$ are unknown parameters. For this choice of parameters, the system exhibits a stable limit cycle, which must be fully contained in the forward-invariant set. This makes it a suitable test case to highlight the key features of the proposed approach. We assume that the observation function $g(\bm{x},u)=x_1$ and the distribution of the measurement noise $w_t \sim \mathcal{N}(0,0.1)$ are known. 

We generate $T=2000$ output measurements $y_{1:T}$ by simulating the system forward from an initial state drawn from a known distribution $\bm{x}_{0} \sim \mathcal{N}(\bm{0}, 0.1\cdot\bm{I}_2)$. As a prior for the unknown parameters, we use $p(\bm{\theta}) = \mathcal{N}(\bm{0},0.1 \cdot \bm{I}_3)$. 

To draw samples from the posterior distribution, we use a staged MMH sampler, as outlined in Remark~\ref{rem:staged_sampler}, which adaptively updates the Gaussian proposal distribution based on the empirical covariance of previously accepted samples. The initial covariance of the proposal distribution is set equal to the covariance of the prior for $\bm{\theta}$ and $\bm{x}_0$. In each stage, 10 additional data points are incorporated, and 500 samples are drawn with a burn-in period of $K_b=200$. Following standard procedure, we scale the sample covariance such that the acceptance rate remains around \SI{25}{\percent}, a commonly recommended value that promotes efficient mixing and reduces autocorrelation between samples \cite{gelman1997}.

First, we verify that Assumption~\ref{as:iid_posterior} is satisfied. For this purpose, we draw $K=10^5$ samples from the joint posterior distribution over model parameters and latent state trajectories using the MMH sampler without thinning. The normalized autocorrelation functions (ACFs) of successive samples of the model parameters $\bm{\theta}$ are shown in Figure~\ref{fig:autocorrelation}. For all three parameters, the ACF decays significantly by a lag of 150. Consequently, we apply a thinning procedure with $k_d=150$ in the following experiments to obtain approximately independent samples.

\begin{figure}[t]
    \definecolor{mycolor1}{rgb}{0.92941,0.69412,0.12549}
    \pgfplotsset{width=9cm, compat = 1.18, 
	height = 4.5cm, grid= major, 
	legend cell align = left, ticklabel style = {font=\scriptsize},
	every axis label/.append style={font=\scriptsize},
	legend style = {font=\scriptsize},
    }
    \def\file{data/autocorrelation.txt}
		
    \centering
    \begin{tikzpicture}
	\begin{axis}[
		grid=both,
		xmin=0, xmax=200,
		ymin=-0.1, ymax=1,
		xtick distance=20,
		ytick distance=0.25,
		ylabel=ACF, xlabel=Lag,
		set layers=standard,
		legend style={font=\scriptsize, at={(1,1)},anchor=north east, row sep=2pt},
		ylabel shift = -6 pt]
				
		\addplot[ultra thick,OrangeRed] table[x=lag,y=theta_1]{\file};
		\addlegendentry{$\theta_1$}
				
		\addplot[ultra thick,mycolor1] table[x=lag,y=theta_2]{\file};
		\addlegendentry{$\theta_2$}
				
		\addplot[ultra thick,OliveGreen] table[x=lag,y=theta_3]{\file};
		\addlegendentry{$\theta_3$}
				
        \end{axis}
    \end{tikzpicture}
    \vspace*{-0.8cm}
    \caption{Normalized autocorrelation function (ACF) of successive samples generated by the MMH sampler without thinning. The red, yellow, and green lines correspond to the ACFs of the sampled values for $\theta_1$, $\theta_2$, and $\theta_3$, respectively.}
    \label{fig:autocorrelation}
    \vspace*{-0.5cm}
\end{figure}

Next, we aim to find a set $\mathcal{C}$ that is forward invariant using the approach proposed in Section~\ref{sec:barrier_certificates}. To ensure a robust certificate, the training set size is chosen as $K_{\text{train}}=750$. An equally sized test set ($K_{\text{test}}=750$) serves to evaluate the certificate and to derive probabilistic guarantees at confidence level $1-\beta=\SI{99.9}{\percent}$. In total, $K=1500$ posterior draws are generated using the MMH sampler with thinning.

We impose the constraint $\mathcal{X}_{\text{min}} \subset \mathcal{C} \subseteq \mathcal{X}_{\text{max}}$ where $\mathcal{X}_{\text{min}}$ and $\mathcal{X}_{\text{max}}$ are circular regions with radii $\rho_{\text{min}} = 1$ and $\rho_{\text{max}} = 3$, respectively. This reflects the requirement that all states with $\lVert \bm{x} \rVert > 3$ are inadmissible, while the choice of $\mathcal{X}_{\text{min}}$ ensures that the invariant set contains the initial state with probability greater than \SI{99}{\percent}. To translate these sets into the SOS framework and restrict all constraints to this region, we use the polynomials $r_{\text{min}}(\bm{x}) = \bm{x}^{\mathsf T}\bm{x} - \rho_\text{min}^2$ and $r_{\text{max}}(\bm{x}) = \bm{x}^{\mathsf T}\bm{x} - \rho_\text{max}^2$. The search for suitable functions $h_{\text{min}}(\cdot)$ and $h_{\text{max}}(\cdot)$ is included in the SOS program, as outlined in Remark~\ref{rem:circular_searchspace}. For the barrier function candidate $h(\cdot)$, we select polynomial degrees $\{0,2,4,6\}$. For $q_1$, $q_2$, and $q_3$, we use degrees $\{0,2,4\}$, and for $q_4$ and $q_5$, we fix the degree to 2. We use SOSTOOLS~\cite{prajna2002}, together with the solver MOSEK \cite{andersen2000}, both to solve the SOS program~\eqref{eq:SOS_program} and to validate the resulting barrier function candidate using condition~\eqref{eq:validation}.

To assess the robustness of the proposed method, the process of computing and validating a barrier function is repeated 100 times. In each run, a new dataset $\mathbb{D}$ is generated as described above. The SOS program~\eqref{eq:SOS_program} is feasible in all 100 runs, yielding a barrier function each time. The values of the \SI{99.9}{\percent} confidence lower bound on the probability that the computed sets are forward invariant for the true system, as given by Theorem~\ref{theo:robustness_guarantees}, range between \SI{97.41}{\percent} and \SI{99.08}{\percent}, with a mean of \SI{98.75}{\percent}. In all cases, the computed set is forward invariant for the true system~\eqref{eq:fitzhugh-nagumo}.

\input{data/Safeset_Plot}

Figure~\ref{fig:invariant_set} shows the results of an exemplary run. The corresponding forward-invariant set $\mathcal{C}$ satisfies $\mathcal{X}_{\text{min}} \subset \mathcal{C} \subseteq \mathcal{X}_{\text{max}}$ and fully encloses the system's limit cycle. The corresponding barrier function $h(\cdot)$ in this run is given by
\begin{subequations}
    \begin{align*}
        h(\bm{x}) = 
              0.064x_1^6 + 0.075x_1^5x_2 + 0.051x_1^4x_2^2 + 0.024x_1^3x_2^3\\
              + 0.028x_1^2x_2^4 + 0.030x_1x_2^5 - 0.008x_2^6 - 1.890x_1^4\\
              - 1.128x_1^3x_2 - 1.605x_1^2x_2^2 - 0.913x_1x_2^3 + 0.124x_2^4\\
              + 1.895x_1^2 - 8.120x_1x_2 - 7.135x_2^2 + 9.997
    \end{align*}
\end{subequations}
For comparison, we repeat the same procedure using samples drawn from the prior instead of the posterior distribution. In this case, the SOS program is infeasible in all 100 runs, and no barrier function candidate can be found. These results, summarized in Table~\ref{tab:posterior_prior_comparison}, highlight the importance of reducing model uncertainty through the MMH sampler presented in Section~\ref{sec:MMH}.

\begin{table}[t]
    \centering
    \caption{Comparison of results using posterior vs. prior sampling}
    \label{tab:posterior_prior_comparison}
    \vspace*{-0.25cm}
    \begin{tabular}{lcc}
        \hline
         & Posterior (proposed method) & Prior (baseline) \\
        \hline
        SOS program feasible & 100 / 100 & 0 / 100 \\
        $1-\delta$ (mean $\pm$ std) & \SI{98.75}{\percent} $\pm$  \SI{0.42}{\percent}& -- \\
        Valid for true system & 100 / 100 & -- \\
        \hline
    \end{tabular}
    \vspace*{-0.5cm}
\end{table}

Overall, the results demonstrate that the proposed method can reliably synthesize and validate barrier certificates under significant model uncertainty, using only noisy output data and without requiring access to full-state measurements.

%% file: data/Safeset_Plot.tex
\begin{figure}[t]
    \definecolor{mycolor1}{rgb}{0.92941,0.69412,0.12549}%
    \definecolor{mycolor2}{rgb}{0.85098,0.32549,0.09804}%
    \definecolor{asparagus}{rgb}{0.53, 0.66, 0.42}
    \usetikzlibrary{decorations.markings}
    \pgfplotsset{width=7cm, 
		height = 7cm, grid= none, 
		legend cell align = left, ticklabel style = {font=\scriptsize},
		every axis label/.append style={font=\scriptsize},
		legend style = {font=\scriptsize}}
    \centering
    \begin{tikzpicture}
    \begin{axis}[%
            xmin=-3.2, xmax=3.2,
    	ymin=-3.2, ymax=3.2,
            xtick = {-3,-2,-1,0,1,2,3},
            ytick = {-3,-2,-1,0,1,2,3},
    	ylabel=$x_2$, xlabel=$x_1$,
    	set layers=standard, 
            legend style={font=\scriptsize, row sep=-3pt, inner xsep=1pt, inner ysep=1pt},
            ylabel shift = -6 pt, legend entries={$\mathcal{X}_\text{max}$, $\mathcal{X}_\text{min}$, $\mathcal{C}$, Limit cycle}]    
            
            \draw[thick, black, fill = lightgray!70] (axis cs:0,0) circle (3);
            \draw[thick, color=OliveGreen, fill = OliveGreen!20] (axis cs:0,0) circle (1);
        
            \addlegendimage{only marks, mark=o, thick, black}  
            \addlegendimage{only marks, mark=o, thick, OliveGreen}  
        
            \def\file{data/C.txt}	               
            \addplot[ultra thick, OrangeRed] table[x=x_1,y=x_2]{\file};   
            
            \def\file{data/Limitcycle.txt}	            
            \addplot[ultra thick, mycolor1, forget plot, postaction={decorate}, 
                decoration={markings, mark=between positions 0.33 and 1.5 step 0.4 with {\arrow{stealth};}}] 
                table[x=x_1,y=x_2]{\file};
            \addplot[ultra thick, mycolor1] table[x=x_1,y=x_2]{\file};            

            \addplot[-Straight Barb, forget plot, color=OrangeRed, point meta={sqrt((\thisrow{u})^2+(\thisrow{v})^2)}, point meta min=0, quiver={u=\thisrow{u}, v=\thisrow{v}, every arrow/.append style={-{Straight Barb[angle'=18.263, scale={1/1500*\pgfplotspointmetatransformed}]}}}]table[row sep=newline] {%
            x1	x2	u	v          
    1.4942   -2.4586    0.0309    0.0014
    1.5564   -2.4706    0.0128    0.0014
    1.6316   -2.4710   -0.0120    0.0014
    1.6917   -2.4582   -0.0346    0.0015
    1.7519   -2.4317   -0.0597    0.0015
    1.7970   -2.4012   -0.0802    0.0015
    1.8334   -2.3684   -0.0981    0.0015
    1.8722   -2.3247   -0.1183    0.0015
    1.9023   -2.2829   -0.1349    0.0015
    1.9323   -2.2332   -0.1525    0.0015
    1.9557   -2.1880   -0.1668    0.0015
    1.9774   -2.1399   -0.1807    0.0015
    1.9999   -2.0827   -0.1958    0.0015
    2.0204   -2.0226   -0.2102    0.0015
    2.0376   -1.9645   -0.2229    0.0015
    2.0526   -1.9067   -0.2345    0.0015
    2.0672   -1.8421   -0.2463    0.0015
    2.0790   -1.7820   -0.2562    0.0015
    2.0892   -1.7218   -0.2654    0.0015
    2.0979   -1.6617   -0.2737    0.0015
    2.1067   -1.5865   -0.2829    0.0014
    2.1128   -1.5216   -0.2898    0.0014
    2.1176   -1.4511   -0.2963    0.0014
    2.1208   -1.3759   -0.3021    0.0014
    2.1220   -1.3008   -0.3066    0.0014
    2.1212   -1.2256   -0.3099    0.0014
    2.1185   -1.1504   -0.3119    0.0013
    2.1136   -1.0752   -0.3127    0.0013
    2.1082   -1.0150   -0.3123    0.0013
    2.0994   -0.9398   -0.3107    0.0013
    2.0907   -0.8797   -0.3084    0.0013
    2.0805   -0.8195   -0.3053    0.0012
    2.0677   -0.7542   -0.3009    0.0012
    2.0526   -0.6881   -0.2954    0.0012
    2.0376   -0.6296   -0.2896    0.0012
    2.0226   -0.5768   -0.2837    0.0012
    2.0043   -0.5188   -0.2764    0.0011
    1.9835   -0.4586   -0.2681    0.0011
    1.9624   -0.4030   -0.2596    0.0011
    1.9420   -0.3534   -0.2514    0.0011
    1.9173   -0.2977   -0.2417    0.0010
    1.8935   -0.2481   -0.2324    0.0010
    1.8705   -0.2030   -0.2235    0.0010
    1.8421   -0.1511   -0.2129    0.0010
    1.8120   -0.0998   -0.2019    0.0009
    1.7825   -0.0526   -0.1914    0.0009
    1.7525   -0.0075   -0.1811    0.0009
    1.7218    0.0359   -0.1709    0.0009
    1.6917    0.0762   -0.1613    0.0008
    1.6617    0.1144   -0.1520    0.0008
    1.6255    0.1579   -0.1414    0.0008
    1.5865    0.2023   -0.1304    0.0007
    1.5564    0.2347   -0.1224    0.0007
    1.5139    0.2782   -0.1117    0.0007
    1.4812    0.3100   -0.1039    0.0007
    1.4361    0.3517   -0.0939    0.0006
    1.3998    0.3835   -0.0863    0.0006
    1.3609    0.4160   -0.0788    0.0006
    1.3158    0.4518   -0.0707    0.0005
    1.2707    0.4859   -0.0633    0.0005
    1.2256    0.5183   -0.0567    0.0005
    1.1810    0.5489   -0.0507    0.0005
    1.1353    0.5787   -0.0453    0.0004
    1.0902    0.6070   -0.0406    0.0004
    1.0451    0.6341   -0.0365    0.0004
    1.0000    0.6601   -0.0330    0.0003
    0.9549    0.6853   -0.0301    0.0003
    0.9011    0.7143   -0.0272    0.0003
    0.8496    0.7411   -0.0252    0.0002
    0.8045    0.7641   -0.0240    0.0002
    0.7535    0.7895   -0.0232    0.0002
    0.6992    0.8160   -0.0229    0.0001
    0.6541    0.8379   -0.0232    0.0001
    0.5987    0.8647   -0.0240    0.0001
    0.5489    0.8887   -0.0253    0.0001
    0.5038    0.9108   -0.0267    0.0000
    0.4450    0.9398   -0.0291   -0.0000
    0.3985    0.9633   -0.0314   -0.0000
    0.3534    0.9866   -0.0339   -0.0001
    0.3000    1.0150   -0.0371   -0.0001
    0.2481    1.0436   -0.0405   -0.0001
    0.2030    1.0693   -0.0437   -0.0002
    0.1579    1.0960   -0.0471   -0.0002
    0.1128    1.1236   -0.0506   -0.0002
    0.0677    1.1523   -0.0542   -0.0003
    0.0226    1.1821   -0.0580   -0.0003
   -0.0226    1.2130   -0.0618   -0.0003
   -0.0677    1.2451   -0.0656   -0.0003
   -0.1128    1.2785   -0.0695   -0.0004
   -0.1579    1.3131   -0.0734   -0.0004
   -0.1992    1.3459   -0.0769   -0.0004
   -0.2360    1.3759   -0.0799   -0.0005
   -0.2782    1.4113   -0.0834   -0.0005
   -0.3233    1.4504   -0.0870   -0.0005
   -0.3580    1.4812   -0.0897   -0.0005
   -0.3985    1.5180   -0.0927   -0.0006
   -0.4398    1.5564   -0.0956   -0.0006
   -0.4737    1.5885   -0.0978   -0.0006
   -0.5184    1.6316   -0.1005   -0.0007
   -0.5491    1.6617   -0.1023   -0.0007
   -0.5940    1.7063   -0.1045   -0.0007
   -0.6244    1.7368   -0.1059   -0.0007
   -0.6688    1.7820   -0.1076   -0.0008
   -0.6992    1.8131   -0.1085   -0.0008
   -0.7422    1.8571   -0.1095   -0.0008
   -0.7744    1.8903   -0.1100   -0.0009
   -0.8153    1.9323   -0.1103   -0.0009
   -0.8496    1.9676   -0.1102   -0.0009
   -0.8889    2.0075   -0.1097   -0.0009
   -0.9248    2.0437   -0.1089   -0.0010
   -0.9640    2.0827   -0.1075   -0.0010
   -1.0000    2.1178   -0.1059   -0.0010
   -1.0421    2.1579   -0.1034   -0.0011
   -1.0752    2.1886   -0.1010   -0.0011
   -1.1203    2.2289   -0.0972   -0.0011
   -1.1604    2.2632   -0.0931   -0.0011
   -1.1974    2.2932   -0.0887   -0.0012
   -1.2406    2.3262   -0.0829   -0.0012
   -1.2857    2.3578   -0.0759   -0.0012
   -1.3308    2.3864   -0.0680   -0.0013
   -1.3798    2.4135   -0.0583   -0.0013
   -1.4361    2.4392   -0.0457   -0.0013
   -1.4942    2.4586   -0.0309   -0.0014
   -1.5564    2.4706   -0.0128   -0.0014
   -1.6316    2.4710    0.0120   -0.0014
   -1.6917    2.4582    0.0346   -0.0015
   -1.7519    2.4317    0.0597   -0.0015
   -1.7970    2.4012    0.0802   -0.0015
   -1.8334    2.3684    0.0981   -0.0015
   -1.8722    2.3247    0.1183   -0.0015
   -1.9023    2.2829    0.1349   -0.0015
   -1.9323    2.2332    0.1525   -0.0015
   -1.9557    2.1880    0.1668   -0.0015
   -1.9774    2.1399    0.1807   -0.0015
   -1.9999    2.0827    0.1958   -0.0015
   -2.0204    2.0226    0.2102   -0.0015
   -2.0376    1.9645    0.2229   -0.0015
   -2.0526    1.9067    0.2345   -0.0015
   -2.0672    1.8421    0.2463   -0.0015
   -2.0790    1.7820    0.2562   -0.0015
   -2.0892    1.7218    0.2654   -0.0015
   -2.0979    1.6617    0.2737   -0.0015
   -2.1067    1.5865    0.2829   -0.0014
   -2.1128    1.5216    0.2898   -0.0014
   -2.1176    1.4511    0.2963   -0.0014
   -2.1208    1.3759    0.3021   -0.0014
   -2.1220    1.3008    0.3066   -0.0014
   -2.1212    1.2256    0.3099   -0.0014
   -2.1185    1.1504    0.3119   -0.0013
   -2.1136    1.0752    0.3127   -0.0013
   -2.1082    1.0150    0.3123   -0.0013
   -2.0994    0.9398    0.3107   -0.0013
   -2.0907    0.8797    0.3084   -0.0013
   -2.0805    0.8195    0.3053   -0.0012
   -2.0677    0.7542    0.3009   -0.0012
   -2.0526    0.6881    0.2954   -0.0012
   -2.0376    0.6296    0.2896   -0.0012
   -2.0226    0.5768    0.2837   -0.0012
   -2.0043    0.5188    0.2764   -0.0011
   -1.9835    0.4586    0.2681   -0.0011
   -1.9624    0.4030    0.2596   -0.0011
   -1.9420    0.3534    0.2514   -0.0011
   -1.9173    0.2977    0.2417   -0.0010
   -1.8935    0.2481    0.2324   -0.0010
   -1.8705    0.2030    0.2235   -0.0010
   -1.8421    0.1511    0.2129   -0.0010
   -1.8120    0.0998    0.2019   -0.0009
   -1.7825    0.0526    0.1914   -0.0009
   -1.7525    0.0075    0.1811   -0.0009
   -1.7218   -0.0359    0.1709   -0.0009
   -1.6917   -0.0762    0.1613   -0.0008
   -1.6617   -0.1144    0.1520   -0.0008
   -1.6255   -0.1579    0.1414   -0.0008
   -1.5865   -0.2023    0.1304   -0.0007
   -1.5564   -0.2347    0.1224   -0.0007
   -1.5139   -0.2782    0.1117   -0.0007
   -1.4812   -0.3100    0.1039   -0.0007
   -1.4361   -0.3517    0.0939   -0.0006
   -1.3998   -0.3835    0.0863   -0.0006
   -1.3609   -0.4160    0.0788   -0.0006
   -1.3158   -0.4518    0.0707   -0.0005
   -1.2707   -0.4859    0.0633   -0.0005
   -1.2256   -0.5183    0.0567   -0.0005
   -1.1810   -0.5489    0.0507   -0.0005
   -1.1353   -0.5787    0.0453   -0.0004
   -1.0902   -0.6070    0.0406   -0.0004
   -1.0451   -0.6341    0.0365   -0.0004
   -1.0000   -0.6601    0.0330   -0.0003
   -0.9549   -0.6853    0.0301   -0.0003
   -0.9011   -0.7143    0.0272   -0.0003
   -0.8496   -0.7411    0.0252   -0.0002
   -0.8045   -0.7641    0.0240   -0.0002
   -0.7535   -0.7895    0.0232   -0.0002
   -0.6992   -0.8160    0.0229   -0.0001
   -0.6541   -0.8379    0.0232   -0.0001
   -0.5987   -0.8647    0.0240   -0.0001
   -0.5489   -0.8887    0.0253   -0.0001
   -0.5038   -0.9108    0.0267   -0.0000
   -0.4450   -0.9398    0.0291    0.0000
   -0.3985   -0.9633    0.0314    0.0000
   -0.3534   -0.9866    0.0339    0.0001
   -0.3000   -1.0150    0.0371    0.0001
   -0.2481   -1.0436    0.0405    0.0001
   -0.2030   -1.0693    0.0437    0.0002
   -0.1579   -1.0960    0.0471    0.0002
   -0.1128   -1.1236    0.0506    0.0002
   -0.0677   -1.1523    0.0542    0.0003
   -0.0226   -1.1821    0.0580    0.0003
    0.0226   -1.2130    0.0618    0.0003
    0.0677   -1.2451    0.0656    0.0003
    0.1128   -1.2785    0.0695    0.0004
    0.1579   -1.3131    0.0734    0.0004
    0.1992   -1.3459    0.0769    0.0004
    0.2360   -1.3759    0.0799    0.0005
    0.2782   -1.4113    0.0834    0.0005
    0.3233   -1.4504    0.0870    0.0005
    0.3580   -1.4812    0.0897    0.0005
    0.3985   -1.5180    0.0927    0.0006
    0.4398   -1.5564    0.0956    0.0006
    0.4737   -1.5885    0.0978    0.0006
    0.5184   -1.6316    0.1005    0.0007
    0.5491   -1.6617    0.1023    0.0007
    0.5940   -1.7063    0.1045    0.0007
    0.6244   -1.7368    0.1059    0.0007
    0.6688   -1.7820    0.1076    0.0008
    0.6992   -1.8131    0.1085    0.0008
    0.7422   -1.8571    0.1095    0.0008
    0.7744   -1.8903    0.1100    0.0009
    0.8153   -1.9323    0.1103    0.0009
    0.8496   -1.9676    0.1102    0.0009
    0.8889   -2.0075    0.1097    0.0009
    0.9248   -2.0437    0.1089    0.0010
    0.9640   -2.0827    0.1075    0.0010
    1.0000   -2.1178    0.1059    0.0010
    1.0421   -2.1579    0.1034    0.0011
    1.0752   -2.1886    0.1010    0.0011
    1.1203   -2.2289    0.0972    0.0011
    1.1604   -2.2632    0.0931    0.0011
    1.1974   -2.2932    0.0887    0.0012
    1.2406   -2.3262    0.0829    0.0012
    1.2857   -2.3578    0.0759    0.0012
    1.3308   -2.3864    0.0680    0.0013
    1.3798   -2.4135    0.0583    0.0013
    1.4361   -2.4392    0.0457    0.0013
    1.4942   -2.4586    0.0309    0.0014 
            };
    \end{axis}            
    \end{tikzpicture}%
    \vspace{-.4cm}
    \caption{Visualization of the obtained forward-invariant set. The red curve marks the boundary of the set $\mathcal{C}$. Red arrows illustrate the system state one step into the future from selected points on this boundary, indicating inward motion. The circular regions $\mathcal{X}_\text{min}$ and $\mathcal{X}_\text{max}$ are shown in green and grey, respectively. The yellow curve depicts the system’s limit cycle, which lies entirely within $\mathcal{C}$.}
    \label{fig:invariant_set}
\end{figure}

%% file: sections/Conclusion.tex
\section{Conclusion}
\label{sec:conclusion}
This paper presents a novel approach for synthesizing probabilistic barrier certificates for unknown systems with latent states and polynomial dynamics using only noisy output measurements. A Bayesian framework is employed to systematically incorporate prior knowledge while providing rigorous uncertainty quantification, and a targeted marginal Metropolis--Hastings sampler is used to generate samples from the posterior distribution over the dynamics. These samples are then used to formulate the search for a barrier certificate as a sum-of-squares optimization problem, and an independent test set of posterior samples is used to establish finite-sample probabilistic guarantees that the certificate is valid for the true system. While this work addresses autonomous systems, extending the approach to output-feedback control laws presents additional challenges and is left for future investigation.